\documentclass[a4paper,10pt]{amsart}
\usepackage{amsmath}
\usepackage{amssymb}
\usepackage{amsthm}
\usepackage{amsfonts}
\usepackage{color}
\usepackage{graphicx}
\usepackage{graphicx,epstopdf}
\usepackage{mathtools}
\usepackage{mathrsfs}
\usepackage{multirow}
\usepackage{rotating}
\usepackage{enumerate}
\usepackage{enumitem}
\usepackage{epsfig}
\usepackage{setspace}
\usepackage{ragged2e}
\usepackage{blindtext}
\usepackage{lineno}
\usepackage{url}
\usepackage{caption}
\usepackage{subcaption}
\usepackage{comment}
\usepackage{physics}
\usepackage{lscape}
\usepackage{float}
\usepackage{array}
\newcolumntype{C}[1]{>{\centering\arraybackslash}m{#1}}
\usepackage[colorlinks=true,linkcolor=blue,citecolor=blue, urlcolor=blue]
{hyperref}%
\newtheorem{cnj}{Conjecture}
\newtheorem{lemma}{Lemma}
\newtheorem{thm}{Theorem}

\newtheorem*{question*}{Question}
\theoremstyle{definition}

\newtheorem{rem}{Remark}

\usepackage[
backend=biber,
style=numeric-comp,
sorting=none
]{biblatex}
\usepackage{geometry}
\title{An Investigation of Additional Food Models with Generalised Functional Response}
\author{Urvashi Verma}
\address{Department of Mathematics, Iowa State University, Ames, IA, USA}
\email{uverma@iastate.edu}

\author{Kanishka Goyal}
\address{Department of Mathematics, Iowa State University, Ames, IA, USA}
\email{kgoyal@iastate.edu}

\author{Saptarshi Biswas}
\address{Department of Computer Science, Iowa State University; Ames National Laboratory, Ames, IA, USA}
\email{sbiswas@iastate.edu; saptarshi.biswas9@gmail.com}

\author{James I. Lathrop}
\address{Department of Computer Science, Iowa State University, Ames, IA, USA}
\email{jil@iastate.edu}

\author{Rana D. Parshad}
\address{Department of Mathematics, Iowa State University, Ames, IA, USA }
\email{rparshad@iastate.edu}
\date{}

\begin{document}

\maketitle
\begin{abstract}
Additional food sources are often used to improve the effectiveness of predators in controlling pest populations.
However, the non-symmetric structure of additional food predator-prey models can cause certain aspects of their dynamics challenging to analyze. In this work, we study a general class of additional food models and establish conditions under which the coexistence equilibrium is globally stable. We then focus on a Holling type IV functional response with AF and show the existence of a Bogdanov-Takens bifurcation of codimension 3. We also study these models through the lens of deterministic chemical reaction network theory. Our analysis shows that the introduction of additional food increases the deficiency of the underlying reaction network and suggests a possible link between higher deficiency and complex bifurcations.
\end{abstract}
\section{Introduction}
Predator-prey dynamics are ubiquitous in ecological systems. They describe the interaction between a hunting organism, a predator, that hunts and depredates a food source, its prey. These models are typically of Gause type, and consist of ordinary differential equations describing the interactions of a prey density $x$, and predator density $y$, given by, $\frac{dx}{dt} = xg(x) - yf(x) , \  
\frac{dy}{dt} =  \epsilon yf(x)  - \delta y$.
In these models, the functional response $f(x)$, which describes the rate at which predators consume prey, and the numerical response, which describes how predator growth changes with prey density, are closely connected and essentially scale linearly with each other. Classical forms for the functional response $f(x)$ are the Holling Types I, II, and III \cite{Holling_1959, holling1959components,holling1965functional, DAWES201311}. These have therefore been widely used to study population dynamics. Predator-prey models find many applications in applied and theoretical ecology. The application we are concerned with in this manuscript is biological control. Herein a predator is introduced into an ecosystem to depredate, and so control a target pest, its prey \cite{V96, S99}. Many predators may also utilize other food resources than the target prey population. Resources such as nectar, pollen, and honeydew can enhance predator reproduction, longevity, and persistence \cite{wade2008conservation,lundgren2009nutritional,vandekerkhove2010pollen}, and when predators fail to suppress pests, their effectiveness can be improved by providing additional food resources \cite{SV06,T15}. The incorporation of additional food into predator–prey models has therefore attracted considerable attention in the literature \cite{van2001alternative,van2002plants, SP07,SP10,SP11,SPD17,SPV18,S02,sen2015global, verma2026additional}. However, additional food (AF) models introduce more complex interactions between prey consumption and predator growth. In such systems, predator growth may depend not only on prey but also on the quality and quantity of additional food. As a result, the forms of the functional and numerical response of the introduced predator in AF models are not symmetric. This asymmetry makes the global analysis of the coexistence equilibrium more challenging, and many of the standard Lyapunov techniques used for classical predator–prey systems are not directly applicable \cite{hsu2005survey}. While a rich literature exists on the global stability of predator-prey systems \cite{hsu2005survey,harrison1979global, hsu1995, korobeinikov2009stability,ding2020global, gomez2025implications}, relatively few studies have established comparable results for models with additional food \cite{SP07,verma2026additional}. In the current manuscript, we develop a general mathematical framework for analyzing the global stability of the coexistence equilibrium in a generalized class of AF models.

Beyond global stability, it is also important to understand the bifurcation structure of AF models since ecological systems often undergo qualitative changes when key parameters are varied \cite{kuznetsov1998elements, perko2013differential}. Existing studies on AF models have primarily focused on codimension-one bifurcations, namely, saddle-node, transcritical, and Hopf bifurcations to name some \cite{SP07, SPV18, ananth2021influence, ananth2021optimal}, as also summarized in Table~\ref{tab:sum_def}. While higher-codimension bifurcations have only recently been explored in AF mediated predator competition models with Type II functional response \cite{parshad2023additional, verma2026}. These higher-codimension bifurcations play a crucial role in understanding the complex dynamics which includes the existence of limit cycles, homoclinic loop, and multistability. In particular, Bogdanov-Takens (BT) bifurcation acts as an organizing center through which these richer dynamical behaviors arise. AF models with a Type IV functional response, have been investigated a fair amount in the recent literature \cite{ananth2021optimal, prakash2026role}. These responses are non-monotone, and so do not permit the dynamic of pest explosion \cite{parshad2023additional}. However, these models lack a detailed investigation concerning higher codimensional bifurcations. Motivated by this gap, we investigate higher-codimension BT bifurcation in AF model with Type IV functional response and establish the existence of codimension-$3$ degeneracies. Although bifurcation theory provides a lens to investigate rich(er) dynamics, in ODE systems, there is no general approach to ascertaining if and when a higher codimensional bifurcation will occur for a given ODE. The approach in the literature relies on making a series of transformations on the given ODE, so as to convert it into a standard normal form, from which the type of bifurcation is determined. This series of affine transformations often is cumbersome, and may not yield reduction to a standard normal form to begin with. Thus it may be useful in answering this question via alternative techniques. This could point to the existence of richer dynamics, without having to actually make the long drawn out estimates to attempt to prove them - particularly if they do not exist in the first place. To this end, we will describe a related field: Chemical Reaction Network Theory (CRNT) \cite{Shapiro_MassActionGibbsFunction_DCRNT_EquilibriumComputation, Aris1965_DCRNT, Krambeck1970_DCRNT}. 

CRNT studies how chemical species behave within a system of chemical reactions, known as a Chemical Reaction Network (CRN) \cite{Feinberg1972_complex_balalnce_DCRNT, HornJackson1972_complex_balancing}. One widely studied mathematical model in this context is the Deterministic Chemical Reaction Network (DCRN), whose behavior is governed by a system of Ordinary Differential Equations (ODEs) \cite{Bournez2021_survey_analog, EpsteinPojman_CRN_book, computability_dcrn}. Specifically, the ODEs represent the rate equations for the underlying kinetic system of the chemical reactions associated with each species. The rate equation for each species represents the rate of change in the species' concentration over time as the reaction progresses. Earlier works by Horn \cite{Horn1972_deficiency_zero}, Jackson \cite{HornJackson1972_complex_balancing}, and Feinberg \cite{Feinberg1995_deficiency_one, FEINBERG198059_deficiency_one, FEINBERG19872229_deficiency_one} introduced the concept of deficiency, represented as \(\delta\). They also presented significant findings known as the Deficiency Zero and Deficiency One theorems. These deficiency theorems offer a well-defined framework for understanding the stability of certain CRNs based on their deficiency values. Deficiency links the linear-algebraic and graph-theoretic concepts of the underlying reaction network to its dynamics. Intuitively, deficiency highlights the linear independence of the reaction vectors and reflects the network's dynamical complexity. However, for the deficiency theorems to be applicable we require the CRN in question to be of (i) low deficiency (ii) of weakly reversible type.  It is unclear what the link is between a CRN of higher deficiency, and the possible higher codimensional bifurcation structure it might possess. Predator-prey systems fall into exactly this class of CRNs. Thus in the current manuscript we will link the deficiency of the concerned predator-prey systems, when viewed through the lens of CRN technology, to the ambient higher codimensional bifurcations that may be present herein. This is accomplished via a series of computations, that then yield certain conjectures.

The section division of the manuscript is as follows. Section \ref{model_formulation} presents the model formulation, the assumptions on the functional and numerical responses, and the preliminary analysis of the model. In Section \ref{equilibrium_analysis}, we find the equilibrium points and derive conditions for
having one or more coexistence equilibria. We then perform a local stability analysis and provide conditions under which the coexistence equilibrium is globally stable. Section \ref{bifurcation_Section} proves the existence of BT bifurcation of codimension-$3$ for the AF model with Type IV functional response. \textcolor{black}{Section \ref{sec:crn}} introduces the concept of deficiency and computes the deficiency of several prey-predator systems. Finally, Section \ref{disc_conclusion} presents the discussion and conclusions.

\section{Model Formulation}
\label{model_formulation}
 The following \emph{general} model represents the dynamics of an introduced predator population $y(t)$ preying on a target pest population $x(t)$, while also provided with an additional food source of quantity $\xi$. Here, $\gamma, \beta$ and $\delta$ are positive constants representing the carrying capacity, the predators' conversion efficiency, and predators' natural mortality respectively. The quality of the additional food source is measured by $\frac{1}{\alpha}$, where $\alpha>0$.
\begin{equation}
\label{Eqn:1g}
\begin{aligned}
\frac{dx}{dt} &= x\left(1-\frac{x}{\gamma}\right) - f(x,\xi,\alpha) y \\    
\frac{dy}{dt} &=  \beta g(x, \xi,\alpha) y  - \delta y
\end{aligned}
\end{equation}

\begin{enumerate}[label=$C_{\arabic*}:$, ref=$C_{\arabic*}$]
    \item \label{C1} $f,g \in C^1(\mathbb{R}_+^2)$ and $f(x,\xi,\alpha)\ge0$, $g(x,\xi,\alpha)\ge0$, \ $\forall$ $x,\xi\ge0$
    \item \label{C2} $f(0,\xi, \alpha) = 0$, \  $\forall \xi\ge0$
  \item \label{C3} $f_x > 0, \ \forall\ x, \ x < \infty$
\item \label{C4}$
\exists\, M_1, M_2 > 0 \ \text{such that} \ 
f(x,\xi,\alpha) \le M_1, \quad g(x,\xi,\alpha) \le M_2, 
\ \forall\, x,\xi \ge 0
$

\item \label{C5} Let $g(x,\xi,\alpha) = f(x, \xi,\alpha) + g_1(x, \xi,\alpha)$,\\ 
where $f$ is the predators' functional response and $g_1$ is the contribution of additional food 

    \item \label{C6} $g_1(x,\xi,\alpha)\ge0$ 
    \item \label{C7} $g_1(x,0,\alpha) = 0 \quad \forall x \ge 0$
    \item \label{C8} $ \exists\, C(\xi, \alpha) > 0 \ \text{s.t.} \
g_1(x,\xi,\alpha) \le C(\xi, \alpha), \quad \forall x \ge 0$ and
$\lim_{\xi\to0} C(\xi,\alpha)=0, \forall \alpha>0$
    \item  \label{C9} $\frac{\partial g_1}{\partial \xi} \ge 0, \quad 
\forall\, x, \xi \ge 0$
\item $
\beta C(\xi,\alpha)<\delta
$
\label{C10}

\end{enumerate}
\subsection{Positivity \& Boundedness}

\begin{thm}
The model \eqref{Eqn:1g} is positively invariant in $\mathbb{R}_+^2$.
\label{pos_inv_gen}
\end{thm}
\begin{proof}
From the differential equation of $x$ we have, $\dot{x}\rvert_{x=0}=0$ using \ref{C2}, and we also have $\dot{y}\rvert_{y=0}=0$. Therefore, by Theorem $A.4$ in \cite{thieme2018mathematics}, the model \eqref{Eqn:1g} is positive invariant in $\mathbb{R}_+^2$. 

    \label{proof_pos_inv_gen}
\end{proof}

\begin{thm}
Suppose conditions \ref{C1} -  \ref{C10} holds then, all solutions of model \eqref{Eqn:1g} originating in $\mathbb{R}_+^2$ are uniformly bounded within a region $\mathcal{B}$ where, 
$$
\mathcal{B}
=
\left\{
(x,y)\in\mathbb{R}_+^2:
0\le x\le \gamma,\,
0\le x+\frac{1}{\beta}y\le \frac{L}{a}
\right\}
$$ where $a=\delta-\beta C(\xi,\alpha)$.
    \label{bdd_gen}
\end{thm}

\begin{proof}
    \label{proof_bdd_gen}

    From the first equation of \eqref{Eqn:1g}, we see that 
$$
\frac{dx}{dt}
=
x\left(1-\frac{x}{\gamma}\right)
-f(x,\xi,\alpha)y
\le
x\left(1-\frac{x}{\gamma}\right)
$$
It follows that the growth of $x(t)$ is dominated by the logistic equation $\frac{du}{dt}=u\left(1-\frac{u}{\gamma}\right), $ whose positive solutions satisfy $\lim_{t\to\infty}u(t)=\gamma$. Hence, $\limsup_{t\to\infty}x(t)\le \gamma$. Now, let us define
$W(t)=x(t)+\frac{1}{\beta}y(t)
$. On differentiating $W(t)$, we get,
\begin{equation*}
\begin{aligned}
\frac{dW}{dt}
&=
\frac{dx}{dt}
+\frac{1}{\beta}\frac{dy}{dt}\\
&=
x\left(1-\frac{x}{\gamma}\right)
-f(x,\xi,\alpha)y
+g(x,\xi,\alpha)y
-\frac{\delta}{\beta}y
\end{aligned}
\end{equation*}
Using condition \ref{C5}, it follows that,
\begin{equation*}
\begin{aligned}
\frac{dW}{dt}
&=
x\left(1-\frac{x}{\gamma}\right)
+g_1(x,\xi,\alpha)y
-\frac{\delta}{\beta}y\\
\frac{dW}{dt}
+\delta W
&=
x\left(1-\frac{x}{\gamma}\right)
+\delta x
+g_1(x,\xi,\alpha)y
\end{aligned}
\end{equation*}
By condition \ref{C8},
$g_1(x,\xi,\alpha)\le C(\xi,\alpha),
$
and therefore
$$
\frac{dW}{dt}
+\delta W
\le
x\left(1-\frac{x}{\gamma}\right)
+\delta x
+C(\xi,\alpha)y
$$
Since
$
x\left(1-\frac{x}{\gamma}\right)
=
\frac{1}{\gamma}x(\gamma-x)$
and the maximum value of $x(\gamma-x)$ is $\gamma^2/4$, when
$0\le x\le \gamma$, it follows that
$
x\left(1-\frac{x}{\gamma}\right)
\le
\frac{\gamma}{4}.
$
Therefore,
$$
\frac{dW}{dt}
+\bigl(\delta-\beta C(\xi,\alpha)\bigr)W
\le
\frac{\gamma}{4}
+\bigl(\delta-\beta C(\xi,\alpha)\bigr)\gamma
$$
Let
$$
a=\delta-\beta C(\xi,\alpha)>0 \ \text{(from \ref{C10}), }\ \quad
L=\frac{\gamma}{4}+a\gamma
$$
Then,
$$
\frac{dW}{dt}+aW\le L
$$
Thus,
$$
0\le W(t)
\le
\frac{L}{a}
\left(1-e^{-at}\right)
+W(0)e^{-at}
$$
So, as $t \rightarrow \infty$ we have, $
0\le W(t)
\le
\frac{L}{a}
$.
This ensures that system \eqref{Eqn:1g} is dissipative with the asymptotic bound $\frac{L}{a}$. Since
$
W(t)=x(t)+\frac{1}{\beta}y(t)$
all solutions eventually enter and remain in the compact region

$$
\mathcal{B}
=
\left\{
(x,y)\in\mathbb{R}_+^2:
0\le x\le \gamma,\,
0\le x+\frac{1}{\beta}y\le \frac{L}{a}
\right\}
$$

Hence, all positive solutions of \eqref{Eqn:1g} are uniformly bounded in $\mathcal{B}$, a proper subset of $\mathbb R_+^2$.
\end{proof}

\section{Equilibrium Analysis}
\label{equilibrium_analysis}
Model \eqref{Eqn:1g} admits the following equilibrium points,

\begin{itemize}
    \item Trivial equilibrium $E_0$: $(0,0)$
 \item Predator-free equilibrium $E_\gamma$: $(\gamma, 0)$
 
\item Coexistence equilibrium point $E^*$: $(x^*,y^*)$

\end{itemize}
\begin{thm}
If we have 
$\beta\, g_1(0,\xi,\alpha) < \delta < \beta\big(f(\gamma,\xi,\alpha) + g_1(\gamma,\xi,\alpha)\big)
$ then, the model \eqref{Eqn:1g} admits at least one coexistence equilibrium $(x^*,y^*)$ with $$ y^* = \frac{x^*\left(1-\frac{x^*}{\gamma}\right)}{f(x^*, \xi,\alpha)} 
\ \text{and} \ g(x^*, \xi,\alpha) = \frac{\delta}{\beta}. $$
\end{thm}
\begin{proof}
At a coexistence equilibrium $(x^*,y^*)$, we have $x^*>0$, $y^*>0$, and
$\frac{dy}{dt} = y(\beta g(x,\xi,\alpha) - \delta) = 0
$
Thus,
$y^*(\beta g(x^*,\xi,\alpha) - \delta) = 0$
Since $y^*>0$, it follows that
$\beta g(x^*,\xi,\alpha) = \delta$.

Using \ref{C5}, we write
$$f(x^*,\xi,\alpha)+g_1(x^*,\xi,\alpha) = g(x^*,\xi, \alpha) \implies f(x^*,\xi,\alpha)+g_1(x^*,\xi,\alpha) = \frac{\delta}{\beta} $$

 Let us define a function $\phi(x,\xi,\alpha) = \beta(f(x,\xi,\alpha)+g_1(x,\xi,\alpha))$. From \ref{C1}, $\phi$ is continuous on $[0,\gamma]$. We evaluate $\phi$ at the endpoints, At $x=0$, $\phi(0,\xi,\alpha) = \beta(f(0,\xi,\alpha)+g_1(0,\xi,\alpha))$
then from \ref{C2} we have, $\phi(0,\xi,\alpha) = \beta g_1(0,\xi,\alpha)$.  At $x=\gamma$,  $\phi(\gamma,\xi,\alpha) = \beta(f(\gamma,\xi,\alpha)+g_1(\gamma,\xi,\alpha))$.

$$ \text{If} \ 
\beta g_1(0,\xi,\alpha) < \delta < \beta(f(\gamma,\xi,\alpha)+g_1(\gamma,\xi,\alpha)) \ 
\text{then,} \ \phi(0,\xi,\alpha) < \delta < \phi(\gamma,\xi,\alpha)$$

Hence, by the Intermediate Value Theorem (IVT), there exists $x^* \in (0,\gamma)$ such that
$$
\phi(x^*,\xi, \alpha) = \delta \implies f(x^*,\xi,\alpha)+g_1(x^*,\xi,\alpha)= g(x^*,\xi,\alpha) = \dfrac{\delta}{\beta}
\ \text{and,} \ y^* = \dfrac{x^*\left(1-\frac{x^*}{\gamma}\right)}{f(x^*, \xi,\alpha)}$$
Since $x^* \in (0,\gamma)$ and $f(x^*,\xi,\alpha) > 0$ then $y^*>0$, thus proving the existence of coexistence equilibrium.

\end{proof}

\begin{rem}
    
If $g_x(x,\xi,\alpha) > 0$ for all $x \in (0,\gamma)$, then the equation
$g(x,\xi,\alpha)=\frac{\delta}{\beta}$
has at most one solution in $(0,\gamma) $. So, the system admits at most one coexistence equilibrium (see figure \ref{fig:unique}). 
If $g_x(x,\xi,\alpha)$ changes sign on $(0,\gamma)$, then multiple coexistence equilibria may exist (see figure \ref{fig:multiple}). 
\label{Remark:unique}
\end{rem}
\begin{figure}
  \begin{subfigure}{.48\textwidth}
\centering
  \includegraphics[width= 6.6cm, height=5cm]{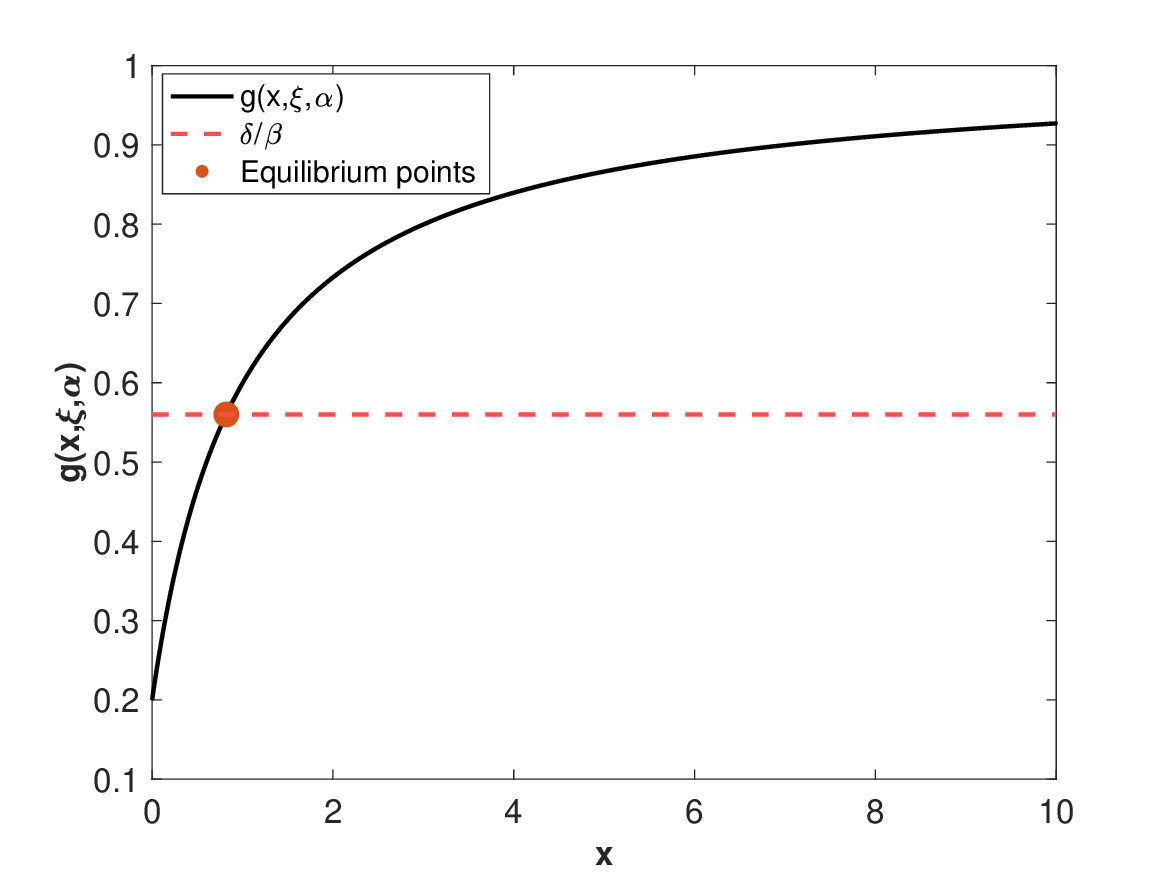}
 \subcaption{Increasing}
 \label{fig:unique}
  \end{subfigure}
  \begin{subfigure}{.48\textwidth}
  \centering
  \includegraphics[width= 6.6cm, height=5cm]{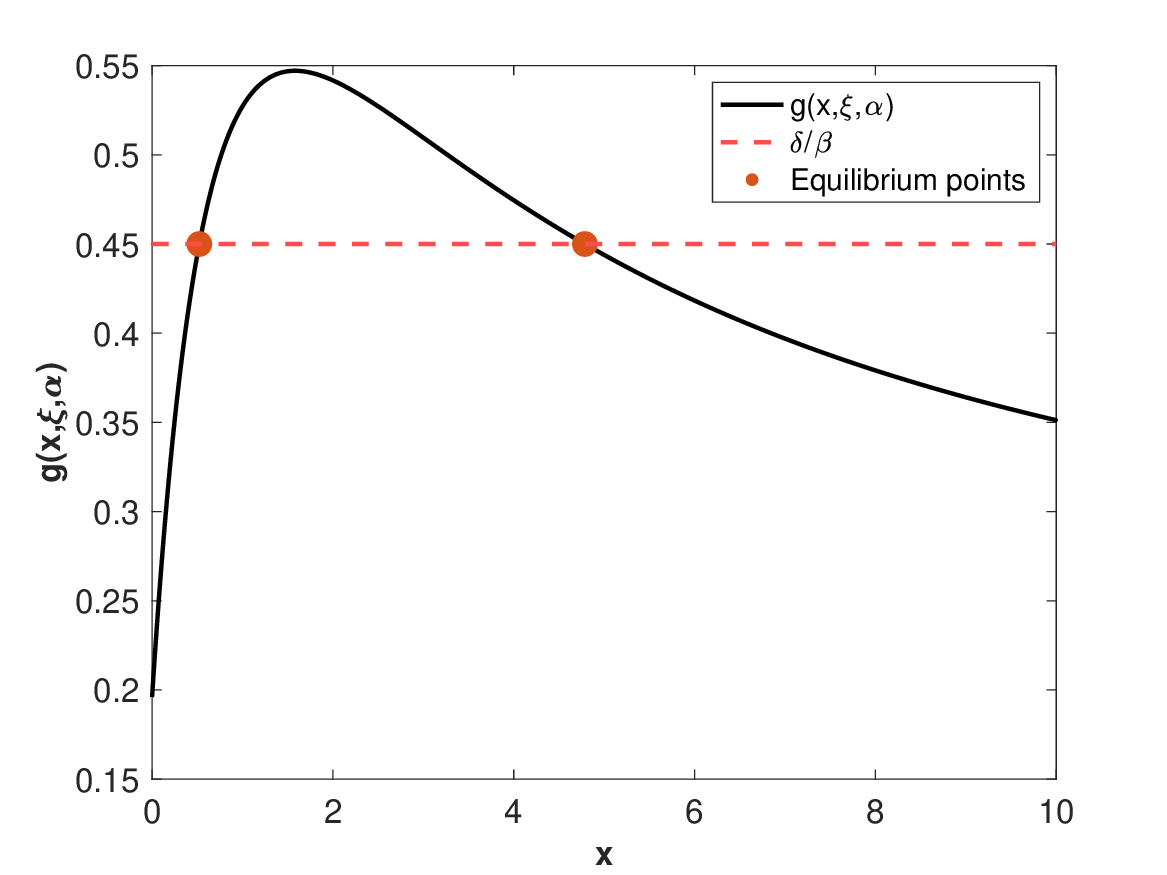}
 \subcaption{Non-monotonic }
\label{fig:multiple}
 \end{subfigure}
 \caption{If the numerical response $g(x, \xi, \alpha)$ is increasing on $0<x<\gamma$, then a unique coexistence equilibrium exist. If the form of the numerical response $g(x, \xi, \alpha)$ is non-monotonic on $0<x<\gamma$, then more than one coexistence equilibrium may exist. }
 \label{exis_equilibrium}
 \end{figure}
We now consider the local stability analysis of the biologically relevant equilibrium points for the system \eqref{Eqn:1g}. The Jacobian matrix $(J)$ for \eqref{Eqn:1g} is given by: 

\begin{equation*}
\centering
    J(x,y)=
\begin{bmatrix}
1 - \frac{2x}{\gamma} - y f_x(x,\xi,\alpha) & -f(x,\xi,\alpha) \\
\beta y g_x(x,\xi,\alpha) & \beta g(x,\xi,\alpha) - \delta
\end{bmatrix}
\end{equation*}

\medskip

\begin{thm}
The trivial equilibrium $E_0$ is unstable, and the equilibrium $E_\gamma$ is locally asymptotically stable if
$\beta g(\gamma,\xi,\alpha) < \delta,$
and is a saddle point if
$\beta g(\gamma,\xi,\alpha) > \delta.$
\label{stability_trivial_equilibrium}
\end{thm}
\begin{proof}
\label{proof_stability_trivial_equilibrium}
At $E_0=(0,0)$
\begin{equation*}
\centering
J(E_0)=
\begin{bmatrix}
1 & -f(0,\xi,\alpha) \\
0 & \beta g(0,\xi,\alpha)-\delta
\end{bmatrix}
\end{equation*}
\medskip

The eigenvalues are $\lambda_1 = 1 > 0$ and $\lambda_2 = \beta g(0,\xi,\alpha)-\delta$, hence $E_0$ is unstable.

\medskip

At $E_\gamma=(\gamma,0)$
\begin{equation*}
J(E_1)=
\begin{bmatrix}
-1 & -f(\gamma,\xi,\alpha) \\
0 & \beta g(\gamma,\xi,\alpha)-\delta
\end{bmatrix}
\end{equation*}
\medskip

The eigenvalues are $\lambda_1 = -1$ and $\lambda_2 = \beta g(\gamma,\xi,\alpha)-\delta$ so, $E_1$ is locally asymptotically stable if
$
\beta g(\gamma,\xi,\alpha) < \delta,
$
and is a saddle point if
$
\beta g(\gamma,\xi,\alpha) > \delta.
$
\end{proof}
\medskip

\begin{thm}
Let $E^*=(x^*,y^*)$ be a coexistence equilibrium of \eqref{Eqn:1g}, where
$\beta g(x^*,\xi,\alpha)=\delta, 
\quad 
y^*=\frac{x^*\left(1-\frac{x^*}{\gamma}\right)}{f(x^*, \ \xi,\ \alpha)}
$ then $E^*$ is, 

\begin{enumerate}
\item[(i)] locally asymptotically stable
if $g_x(x^*,\xi,\alpha)>0$, and $\dfrac{\gamma}{2}<x^*<\gamma$.

\item[(ii)] locally asymptotically stable if $g_x(x^*,\xi,\alpha)>0$, and $\left( \frac{f(x^*,\ \xi,\ \alpha)}{x^*}\right)' \geq 0 $

\item[(iii)] is a saddle if 
 $g_x(x^*,\xi,\alpha)<0$.
\end{enumerate}
    \label{stability_coexistence_equilibrium}
\end{thm}
\begin{proof}
At $E^*=(x^*,y^*)$,
using $\beta g(x^*,\xi,\alpha)=\delta$, the Jacobian matrix  becomes,

$$
J(E^*) =
\begin{bmatrix}
1 - \frac{2x^*}{\gamma} - y^* f_x(x^*,\xi,\alpha) & -f(x^*,\xi,\alpha) \\
\beta y^* g_x(x^*,\xi,\alpha) & 0
\end{bmatrix}
$$

The characteristic polynomial is given by 
$
\lambda^2 - \text{Tr}(J(E^*))\,\lambda + \det(J(E^*)) = 0,
$
and after substituting $y^* = \frac{x^*\left(1-\frac{x^*}{\gamma}\right)}{f(x^*,\ \xi,\ \alpha)}$ we have, 
\begin{equation}
    \text{Tr}(J(E^*)) =
\left(1 - \frac{2x^*}{\gamma}\right)
- x^*\left(1-\frac{x^*}{\gamma}\right)\frac{f_x(x^*, \ \xi,\alpha)}{f(x^*,\xi,\alpha)}, \quad \det(J(E^*)) =
\beta x^*\left(1-\frac{x^*}{\gamma}\right) g_x(x^*,\xi,\alpha)
\label{tr_det_int}
\end{equation}

By the Routh--Hurwitz criterion, the equilibrium $E^*$ is locally asymptotically stable if and only if 
$$ \text{Tr}(J(E^*)) < 0
\ \text{and} \
\det(J(E^*)) > 0$$

Since $x^* \in (0,\gamma)$ and $\beta>0$, we have
$$
\det(J(E^*)) > 0 \iff g_x(x^*,\xi,\alpha) > 0
$$
Moreover, if $\frac{\gamma}{2} < x^* < \gamma$ and using \ref{C3} ensures $f_x(x^*,\xi,\alpha)>0$, we have $\text{Tr}(J(E^*)) < 0$. Now, rewriting the trace from equation \eqref{tr_det_int}, using $\frac{y^* f(x^*,\ \xi,\ \alpha)}{x^*} = 1-\frac{x^*}{\gamma}$ we have, 

$$ \text{Tr}(J(E^*)) = \frac{y^* f(x^*,\ \xi,\ \alpha)}{x^*}-
\frac{x^*}{\gamma}
- y^*  f_x(x^*, \ \xi,\alpha) = x^*y^* \left( \frac{f(x^*,\ \xi,\ \alpha)}{(x^*)^2} - \frac{f_x(x^*, \ \xi,\alpha)}{x^*}\right)-\frac{x^*}{\gamma} $$

$$ \text{Tr}(J(E^*)) = -x^*y^* \left( \frac{x^*f_x(x^*, \ \xi,\alpha) - f(x^*,\ \xi,\ \alpha)}{(x^*)^2} \right) -\frac{x^*}{\gamma} = -x^*y^* \left( \frac{f(x^*,\ \xi,\ \alpha)}{x^*}\right)'-\frac{x^*}{\gamma}  $$

Therefore if we have $\left( \frac{f(x^*,\ \xi,\ \alpha)}{x^*}\right)' \geq 0 $ then the trace will be negative and the equilibrium $E^*$ will be locally asymptotically stable. 
But if $g_x(x^*,\xi,\alpha) < 0$, then $\det(J(E^*)) < 0$, and the eigenvalues are real with opposite signs. Therefore, $E^*$ is a saddle point and hence unstable.

\label{proof_stability_coexistence_equilibrium}
\end{proof}

\subsection{Global Stability Analysis}

The equations representing the general model are given in \eqref{Eqn:1g}, but from \ref{C5} we have, $g(x,\xi,\alpha) = f(x,\xi,\alpha) + g_1(x,\xi,\alpha)$, then the predator equation can be written as, 
\begin{equation}
\frac{dy}{dt} = \beta f(x,\xi,\alpha) y + \beta g_1(x,\xi,\alpha) y - \delta y
\label{eq:predator_decomp}
\end{equation}

From \ref{C8}, there exists $C(\xi,\alpha) > 0$ such that, $g_1(x,\xi,\alpha) \le C(\xi,\alpha), \ \forall x \ge 0 $, then from \eqref{eq:predator_decomp}, we have 
\begin{equation}
 \frac{dy}{dt} \le \beta f(x,\xi,\alpha) y + \beta C(\xi,\alpha) y - \delta y 
\label{eq:predator_super_Comp}
 \end{equation}
Let $\tilde \delta = \delta - \beta C(\xi,\alpha) $, then the model \eqref{Eqn:1g} has a   super solution (in the predator component) given by \eqref{eq:super_system}:

\begin{equation}
\begin{aligned}
\frac{d\tilde{x}}{dt} &= \tilde{x}\left(1-\frac{\tilde{x}}{\gamma}\right) - f(\tilde{x},\xi,\alpha)\tilde{y} \\    
\frac{d\tilde{y}}{dt} &= \beta f(\tilde{x},\xi,\alpha)\tilde{y} - \tilde{\delta} \tilde{y}
\end{aligned}
\label{eq:super_system}
\end{equation}

From equation \eqref{eq:predator_super_Comp}, the comparison principle implies that $y(t)\le \tilde{y}(t), \forall t\ge 0$
where $\tilde{y}(t)$ denotes the solution of predator equation in system \eqref{eq:super_system}. We can state the following lemma about system \eqref{eq:super_system},

\begin{lemma}
If we have
$0 < \tilde{\delta} < \beta f(\gamma,\xi,\alpha)$
then the system \eqref{eq:super_system} admits at least one coexistence equilibrium 
$(\tilde{x}^*, \tilde{y}^*)$ with
$f(\tilde{x}^*,\xi,\alpha) = \frac{\tilde{\delta}}{\beta}, 
\quad \tilde{y}^* = \frac{\tilde{x}^*\left(1-\frac{\tilde{x}^*}{\gamma}\right)}{f(\tilde{x}^*,\ \xi, \ \alpha)}.$
\label{coexistence_existence_super_1}
\end{lemma}
\begin{proof}
    At a coexistence equilibrium $(\tilde{x}^*,\tilde{y}^*)$, we have $\tilde{x}^*>0$, $\tilde{y}^*>0$, and $\frac{d\tilde{y}}{dt} = \tilde{y}(\beta f(\tilde{x},\xi,\alpha) - \tilde{\delta}) = 0$.
Thus, $\tilde{y}^*(\beta f(\tilde{x}^*,\xi,\alpha) - \tilde{\delta}) = 0$.
Since $\tilde{y}^*>0$, it follows that
$\beta f(\tilde{x}^*,\xi,\alpha) = \tilde{\delta}$. Now, let us define $\phi(x,\xi,\alpha) = \beta f(x,\xi,\alpha)$. From \ref{C1}, $\phi$ is continuous on $[0,\gamma]$. Using \ref{C2}, we have
$
\phi(0,\xi,\alpha) = 0, \quad \phi(\gamma,\xi,\alpha) = \beta f(\gamma,\xi,\alpha)
$. If
$0 < \tilde{\delta} < \beta f(\gamma,\xi,\alpha),
\implies
\phi(0,\xi,\alpha) < \tilde{\delta} < \phi(\gamma,\xi,\alpha).
$ Hence by IVT, there exists $x^* \in (0,\gamma)$ such that
$$\phi(x^*,\xi,\alpha) = \tilde{\delta} \implies f(x^*,\xi,\alpha) = \frac{\tilde{\delta}}{\beta}.
$$ Thus, setting $\tilde{x}^* = x^*$, we obtain, 
$\tilde{y}^* = \frac{\tilde{x}^*\left(1-\frac{\tilde{x}^*}{\gamma}\right)}{f(\tilde{x}^*,\ \xi, \ \alpha)}.$ Since $\tilde{x}^* \in (0,\gamma)$ and $f(\tilde{x}^*,\xi, \alpha) > 0$, we have $\tilde{y}^*>0$. 
Thus, the system admits a coexistence equilibrium.
\label{proof_coexistence_existence_super_1}
\end{proof}

\begin{rem}
Since from \ref{C3} we have, $\frac{\partial f}{\partial x} > 0$ for all $x>0$, the equation 
$f(x,\xi, \alpha) = \frac{\tilde{\delta}}{\beta}$ admits a unique solution. 
Hence, the coexistence equilibrium is unique.
\label{rem:uniqueness}
\end{rem}

\begin{lemma}
\label{lem:p1s_1} Consider system \eqref{eq:super_system}. If 
$0 < \tilde{\delta} < \beta f(\gamma,\xi,\alpha)
$ and the function $H(x)=\frac{x(1-x/\gamma)}{f(x,\xi,\alpha)}$
is strictly decreasing on $(0,\gamma)$ then the coexistence equilibrium  \ $ (\tilde {x^*}, \tilde {y^*})$, is globally asymptotically stable.
\end{lemma}
\begin{proof}
We rewrite the system as
$$\dot{\tilde{x}} = \tilde{x}\left(1-\frac{\tilde{x}}{\gamma}\right) - f(\tilde{x},\xi,\alpha)\tilde{y}, 
\quad
\dot{\tilde{y}} = (\beta f(\tilde{x},\xi,\alpha) - \tilde{\delta})\tilde{y}
$$
Consider the following Lyapunov function, \cite{hsu2005survey}
\begin{equation}
V = \int_{\tilde{x}^*}^{\tilde{x}} \frac{f(\phi,\xi,\alpha) - f(\tilde{x}^*,\xi,\alpha)}{f(\phi,\xi,\alpha)}\, d\phi 
+ \frac{1}{\beta}\int_{\tilde{y}^*}^{\tilde{y}} \frac{\eta - \tilde{y}^*}{\eta}\, d\eta
 \label{lyapunov_construction_sup_gen}
\end{equation}

Differentiating $V$ with respect to time $t$ along the solutions of model \eqref{eq:super_system}, we get
\begin{equation*}
    \dot{V}
= \frac{f(\tilde{x},\xi,\alpha) - f(\tilde{x}^*,\xi,\alpha)}{f(\tilde{x},\xi,\alpha)}\, \dot{\tilde{x}}
+ \frac{1}{\beta} \frac{\tilde{y} - \tilde{y}^*}{\tilde{y}}\, \dot{\tilde{y}}
\end{equation*}

Using system of equations \eqref{eq:super_system} we have,

\begin{equation*}
\setlength{\jot}{10pt}
\begin{aligned}
\dot{V}
&= \frac{f(\tilde{x},\xi,\alpha) - f(\tilde{x}^*,\xi,\alpha)}{f(\tilde{x},\xi,\alpha)}
\left(\tilde{x}\left(1-\frac{\tilde{x}}{\gamma}\right) - f(\tilde{x},\xi,\alpha)\tilde{y}\right) + (\tilde{y} - \tilde{y}^*) \left(f(\tilde{x},\xi,\alpha) - f(\tilde{x}^*,\xi,\alpha)\right) \\
&= \left(f(\tilde{x},\xi,\alpha) - f(\tilde{x}^*,\xi,\alpha)\right)
\left( \frac{\tilde{x}\left(1-\frac{\tilde{x}}{\gamma}\right)}{f(\tilde{x},\xi,\alpha)} - \tilde{y}^* \right)
\end{aligned}
\end{equation*}

Since $f_x>0$ (see \ref{C3}), therefore $
f(\tilde{x},\xi,\alpha) - f(\tilde{x}^*,\xi,\alpha)$ is negative if $\tilde{x} < \tilde{x}^*$ and positive if $\tilde{x} > \tilde{x}^*$. Also by the monotonicity of $H(x)$, we have $\frac{\tilde{x}(1-\tilde{x}/\gamma)}{f(\tilde{x},\xi,\alpha)} - \tilde{y}^*$ is positive  if $\tilde{x} < \tilde{x}^*$ and negative if  $\tilde{x} > \tilde{x}^*$.

Therefore, we have that 

\begin{equation*}
     \dot{V}\rvert_{({\tilde{x_1} ,{\tilde{y_1})} = (\tilde{x^*_1},\tilde{y^*_1})}}  \leq 0 \ \text{in}  \ \mathbb{R}^2_{+}
\end{equation*}
Thus, the lemma is proved.
    \label{proof_lem:p1s_1}
\end{proof}

Now using the assumption \ref{C5}, and equation \eqref{eq:predator_decomp}, 
since $g_1(x,\xi,\alpha) \ge 0$ for all $x,\xi \ge 0$, we obtain the lower bound
\begin{equation}
\frac{dy}{dt} \ge \beta f(x,\xi,\alpha) y - \delta y
\label{sub_comp_predator}
\end{equation}
Let $\Bar{\delta} = \delta$, then the model \eqref{Eqn:1g} has the following sub solution (in the predator component) given by \eqref{eq:sub_system}:

\begin{equation}
\begin{aligned}
\frac{d\Bar{x}}{dt} &= \Bar{x}\left(1-\frac{\Bar{x}}{\gamma}\right) - f(\Bar{x},\xi,\alpha)\Bar{y} \\    
\frac{d\Bar{y}}{dt} &= \beta f(\Bar{x},\xi,\alpha)\Bar{y} - \Bar{\delta} \Bar{y}
\end{aligned}
\label{eq:sub_system}
\end{equation}

From equation \eqref{sub_comp_predator}, the comparison principle implies that $\bar{y}(t) \le y(t), \forall t\ge 0$
where $\bar{y}(t)$ denotes the solution of predator equation in system \eqref{eq:sub_system}. We can state the following lemma about system \eqref{eq:sub_system},
 \begin{lemma}
If
$0 < \Bar{\delta} < \beta f(\gamma,\xi,\alpha)$. 
Then the system \eqref{eq:sub_system} admits at least one coexistence equilibrium 
$(\Bar{x}^*, \Bar{y}^*)$ with
$
f(\Bar{x}^*,\xi,\alpha) = \frac{\Bar{\delta}}{\beta}, 
\quad 
\Bar{y}^* = \frac{\Bar{x}^*\left(1-\frac{\Bar{x}^*}{\gamma}\right)}{f(\Bar{x}^*,\ \xi, \ \alpha)}.
$
  \label{coexistence_existence_sub_1}
\end{lemma}
\begin{proof}
The proof is identical to that of Lemma~\ref{coexistence_existence_super_1}. Particularly, if we define $\phi(x,\xi,\alpha)=\beta f(x,\xi,\alpha)$ and apply IVT
on $[0,\gamma]$ with $\Bar{\delta}$ in place of $\tilde{\delta}$ to obtain 
$\Bar{x}^* \in (0,\gamma)$ such that $f(\Bar{x}^*,\xi,\alpha)=\Bar{\delta}/\beta$. 
The expression for $\Bar{y}^*$ then follows similarly and since $\Bar{x}^* \in (0,\gamma)$ and $f(\Bar{x}^*,\xi,\alpha) > 0$, it follows that $\Bar{y}^*>0$. Thus, the system admits a coexistence equilibrium, and the uniqueness of the coexistence equilibrium follows as mentioned in remark \ref{rem:uniqueness}.
\label{proof_coexistence_existence_sub_1}
\end{proof}
\begin{lemma}
\label{lem:p2s_1}
Consider system \eqref{eq:sub_system}. If
$0 < \Bar{\delta} < \beta f(\gamma,\xi, \alpha)
$ and the function $H(x)=\frac{x(1-x/\gamma)}{f(x,\xi, \alpha)}$
is strictly decreasing on $(0,\gamma)$ then, the coexistence equilibrium $(\Bar{x}^*, \Bar{y}^*)$, is globally asymptotically stable.
\end{lemma}
\begin{proof}
The proof follows along the same lines as shown in the proof of Lemma \ref{lem:p1s_1}. We consider the Lyapunov function (see equation \ref{lyapunov_construction_sup_gen}) and, under the assumption that $H(x)$ is strictly decreasing on $(0,\gamma)$ we will have $\dot{V} = \left(f(\Bar{x},\xi,\alpha)-f(\Bar{x}^*,\xi, \alpha)\right)
\left( \frac{\Bar{x}\left(1-\frac{\Bar{x}}{\gamma}\right)}{f(\Bar{x},\xi, \alpha)} - \Bar{y}^* \right) \le 0$, proving the coexistence equilibrium  $(\Bar{x}^*,\Bar{y}^*)$ is globally asymptotically stable.
  \label{proof_lem:p2s_1}  
\end{proof}
We next state the following auxiliary theorem from \cite{perko2013differential},

\begin{thm}
\label{thm:dp1} (Dependence on parameters)

Let $E$ be an open subset of $R^{n+m}$ containing the point $(\bf{x_0}, \bf{\mu_0})$ where $(\bf{x_0}) \in R^m$ and assume that $\bf{f} \in C^1(E)$. It then follows that there exists $a>0$, $\delta> 0 $ such that for all $\bf{y} \in N_{\delta}(x_0)$ and $\bf{\mu} \in N_{\delta}(\mu_0)$, the initial value problem,

\begin{equation*}
    \dot{\bf{x}} = \bf{f}(\bf{x}, \bf{\mu}), \qquad
    \bf{x}(0) = \bf{y}
\end{equation*}

has a unique solution $\bf{u}(t,\bf{y},\bf{\mu})$ with $\bf{u}$ $\in C^1(G)$ where $G = [-a,a] \times N_{\delta}(x_0) \times N_{\delta}(\mu_0)$.
\end{thm}
\begin{thm}
\label{thm:gs1_general} 
Consider system \eqref{Eqn:1g}. If $H(x)=\frac{x(1-x/\gamma)}{f(x,\xi, \alpha)}$
is strictly decreasing on $(0,\gamma)$. Then, under the parametric restrictions, $
0 < \tilde{\delta} < \beta f(\gamma,\xi, \alpha)
$ we have the solution $ (x,y)$ to \eqref{Eqn:1g} is globally asymptotically stable.
\end{thm}
\begin{proof}
\label{proof_gs}
We proceed by contradiction. Without loss of generality, assume that $(x,y)$ is not globally asymptotically stable, that is, solutions with sufficiently large initial data do not converge to the $(x^{*}, y^{*})$ state. Under the conditions needed to prove Lemma \ref{lem:p1s_1} and \ref{lem:p2s_1}, where  the coexistence equilibrium 
$(\tilde{x_1},\tilde{y_1})$ for system \eqref{eq:super_system} and $(\Bar {{x_1}},\Bar {{y_1}})$ for system \eqref{eq:sub_system} are globally asymptotically stable. By direct comparison, $\Bar {{y_1}} \leq y_{1} \leq \tilde{y_{1}}$. We next choose a sequence of 
parameter values
$\{(\xi_n,\alpha_n)\}_{n=1}^{\infty}$ such that, for each $n$,
 we have the global stability of $(\tilde{x_1},\tilde{y_1})$ and $(\Bar{x_1},\Bar{y_1})$, and furthermore,

\begin{equation*}
    \lim_{n \rightarrow \infty} \xi_{n} \alpha_{n} \rightarrow \alpha \xi, \ \lim_{n \rightarrow \infty} \xi_{n}  \rightarrow 0
\end{equation*}

Since the functional response $f(x,\xi,\alpha)$ depends on the parameters
$\xi$ and $\alpha$ only through the product $\alpha\xi$, it follows that 
$f(x,\xi_n,\alpha_n)=f(x,\xi,\alpha), \forall n$ and since $\xi_n\to0 $, from assumption \ref{C7}-\ref{C8} we have, $g_1(x,\xi_n,\alpha_n)\to0,
\quad C(\xi_n,\alpha_n)\to0$.
Now, since the dependence of $(\tilde{x_1},\tilde{y_1})$ and $(\Bar{x_1},\Bar{y_1})$, on the parameters $\xi, \alpha$, is smooth, via standard $C^{1}$ theory of dynamical systems, with respect to parameters via Theorem \ref{thm:dp1}, we can consider a sequence of solutions, $(\tilde{x^{n}_1},\tilde{y^{n}_1})$ and $(\Bar{x^{n}_1},\Bar{y^{n}_1})$, that have the parameters $\{\xi, \alpha \}$ replaced by $\{\xi_{n}, \alpha_{n} \}$. Via comparison we have,

\begin{equation}
   \Bar {{y_1}} \leq \Bar {{y^{n}_1}} \leq y_{1} \leq \tilde{y^{n}_1} \leq \tilde{y_1}
\end{equation}

Clearly, by the construction we have that in the limit,

\begin{equation}
   \Bar {{y_1}} = \lim_{n \rightarrow \infty} \Bar {{y^{n}_1}} = \lim_{n \rightarrow \infty} \tilde{y^{n}_1} 
\end{equation}

Thus we take limits to yield,

\begin{equation}
   \Bar {{y_1}} = \lim_{n \rightarrow \infty} \Bar {{y^{n}_1}} \leq y_{1} \leq \lim_{n \rightarrow \infty} \tilde{y^{n}_1} = \lim_{n \rightarrow \infty} \Bar {{y^{n}_1}} = \Bar {{y_1}}
\end{equation}

It also follows that,

\begin{equation}
   \Bar {{x_1}} \geq \Bar {{x^{n}_1}} \geq x_{1} \geq \tilde{x^{n}_1} \geq \tilde{x_1}
\end{equation}

Thus, taking limits via the same argument as earlier,

\begin{equation}
   \Bar {{x_1}} = \lim_{n \rightarrow \infty} \Bar {{x^{n}_1}} \geq x_{1} \geq \lim_{n \rightarrow \infty} \tilde{x^{n}_1} = \lim_{n \rightarrow \infty} \Bar {{x^{n}_1}} = \Bar {{x_1}}
\end{equation}

Thus $(x_{1},y_{1})$ will be driven to the globally stable $(\Bar {{x^{*}_1}}, \Bar {{y^{*}_1}})$ state, which can be made arbitrarily close to the $(x^{*}_{1}, y^{*}_{1})$ state - this yields a contradiction, and the theorem is proved.

\end{proof}
\section{Bifurcations for type IV Functional Response}
\label{bifurcation_Section}
For the model \eqref{Eqn:1g}, we study BT bifurcation with type IV Functional response, then we rewrite system \eqref{Eqn:1g} as
\begin{equation}
\label{Eqn:model_type4}
\begin{aligned}
\frac{dx}{dt} &= f(x,\xi)\left(h(x,\xi)-y\right), \\    
\frac{dy}{dt} &=  y\left(\beta g(x,\xi)  - \delta\right).
\end{aligned}
\end{equation}
with
\begin{equation}
\label{Eqn:2BT}
\begin{aligned}
f(x,\xi)&=\frac{x}{(1+\alpha\xi)(bx^2+1)+x}, \quad    
g(x,\xi)= \frac{x+\xi(bx^2+1)}{(1+\alpha\xi)(bx^2+1)+x},\\
h(x,\xi)&=\left(1-\frac{x}{\gamma}\right)\left((1+\alpha\xi)(bx^2+1)+x\right).
\end{aligned}
\end{equation}
Or equivalently,
\begin{equation}
\label{Eqn:type4_simp}
\begin{aligned}
\frac{dx}{dt} &= f(x,\xi)\left(h(x,\xi)-y\right), \\    
\frac{dy}{dt} &=  \beta y\left( g(x,\xi)  - d\right).
\end{aligned}
\end{equation}
where $d=\frac{\delta}{\beta}$

\subsection{Linear Stability Analysis}
\subsubsection{Equilibrium Points}
The system admits two boundary equilibria: $E_0=(0,0)$, extinction of both species, and $E_\gamma=(\gamma,0)$, extinction of predators.

For the interior equilibrium points, from the prey nullcline, we have $y^*=h(x^*,\xi)$ and from the predator nullcline,
\begin{equation}
    g(x,\xi)=\frac{x+\xi(bx^2+1)}{(1+\alpha\xi)(bx^2+1)+x}=d
\end{equation}
\begin{equation}
     b( \xi -d (1+\alpha \xi ))x^2 +(1-d) x -d (1+\alpha \xi ) + \xi=0
\end{equation}
Which yields two roots,
\begin{equation}
    x_1^*=\frac{(1-d)+\sqrt{(d-1)^2-4 b (d (1+\alpha \xi )-\xi )^2}}{2 b (d (1+\alpha \xi )-\xi )}, \quad
    x_2^*=\frac{(1-d)-\sqrt{(d-1)^2-4 b (d (1+\alpha \xi )-\xi )^2}}{2 b (d (1+\alpha \xi )-\xi )},
\end{equation}
this gives us two possible equilibria: $E_{x_1^*}=(x_1^*,h(x_1^*,\xi))$ and $E_{x_2^*}=(x_2^*,h(x_2^*,\xi))$. The two equilibria coalesce if the discriminant is zero. That is,
\begin{equation}\label{discriminant}
    \Delta=d^2 \left(1-4 b (1+\alpha \xi )^2\right)+d (8 b \xi  (1+\alpha \xi )-2)-4 b \xi ^2+1=0
\end{equation}
Solving this for $d$ gives,
\begin{equation}\label{Eqn:d_star}
    d=\frac{1+2\xi \sqrt{b}}{1+2 \sqrt{b}(1+\alpha \xi )}:=d^*
\end{equation}
where $\left(1+2 \sqrt{b} (1+\alpha  \xi)\right)\neq0$ or $\sqrt{b}\neq-\frac{1}{2(1+\alpha \xi )}$.\\
$x_1^*$ and $x_2^*$ exist when $d\in(0,d^*)$ and they coalesce when $d=d^*$ at $x=\frac{1}{\sqrt{b}}$.

\subsubsection{Linear Analysis}
The Jacobian matrix for \eqref{Eqn:type4_simp} at any equilibrium point $(x,y)$ is given by:
\begin{equation} \label{vm}
J(x,y)=
\begin{bmatrix}
f(x,\xi)h_x(x,\xi) - f_x(x,\xi)(h(x,\xi)-y) & -f(x,\xi) \\
\beta y g_x(x,\xi) & \beta \left(g(x,\xi) - d\right)
\end{bmatrix},
\end{equation}
In particular, at the interior equilibrium $E^* = (x^*, y^*)$ with $y^* = h(x^*,\xi)$, we have
\begin{equation} \label{vm_in}
J(x^*,y^*)=
\begin{bmatrix}
f(x^*,\xi)h_x(x^*,\xi) & -f(x^*,\xi) \\
\beta h(x^*,\xi) g_x(x^*,\xi) & 0
\end{bmatrix},
\end{equation}
At this interior equilibrium,
\begin{itemize}
    \item $\operatorname{tr}\left(J(x^*, h(x^*,\xi))\right)=$ $f(x^*,\xi)h_x(x^*,\xi)$,
    \vskip 3mm
    \item $\det\left(J(x^*, h(x^*,\xi))\right)=\beta h(x^*,\xi)f(x^*,\xi) g_x(x^*,\xi)$.
\end{itemize}

\subsection{Bogdanov-Takens Bifurcation}
For the system \eqref{Eqn:type4_simp} to have nilpotent singularity, it should have double-zero eigenvalues. From the linear analysis above, it can be noted that this system is nilpotent if $h_x(x_1^*,\xi)=0$ and $g_x(x_1^*,\xi)=0$.

From $g_x(x_1^*,\xi)=0$, we get
\begin{equation}
\begin{aligned}
x_1^*&=x_2^*=\frac{1}{\sqrt{b}}:=x^*,\\
     E_{x_1^*}&=E_{x_2^*}=\left(\frac{1}{\sqrt{b}},h\left(\frac{1}{\sqrt{b}},\xi\right)\right):=E^{**}
\end{aligned}
\end{equation}
Also, from \eqref{Eqn:d_star} when $E_{x_1^*}=E_{x_2^*}$, we have
\begin{equation}
    d=\frac{1+2\xi \sqrt{b}}{1+2 \sqrt{b}(1+\alpha \xi )}:=d^*
\end{equation}

From $h_x(x_1^*,\xi)=0$, 
\begin{equation}
    h_x\left(\frac{1}{\sqrt{b}},\xi\right)=\frac{\left(\gamma\sqrt{b}  -2\right) \left(1+2 \sqrt{b} (1+\alpha  \xi)\right)}{\gamma\sqrt{b}}=0.
\end{equation}
Since, one cannot have $\left(1+2 \sqrt{b} (1+\alpha  \xi)\right)=0$, solving this equation for $\gamma$ gives us the critical value
\begin{equation}
    \gamma=\frac{2}{\sqrt{b}}:=\gamma^*
\end{equation}
\begin{thm}
    When $(d,\gamma)=(d^*,\gamma^*)$, $E^{**}=(x^*, h(x^*,\xi))$ is a cusp singularity of codimension-$2$ for any \( x^* > 0 \), \( \alpha > 0 \), \( \xi > 0 \). Moreover, if
    \begin{equation}
        (d,\gamma,b)=\left(\frac{1}{3}\left(1+\frac{2\xi}{(1+\alpha\xi)}\right),2(1+\alpha\xi),\frac{1}{(1+\alpha\xi)^2}\right):=(d^*,\gamma^*,b^*),
    \end{equation}
    then $E^{**}=(x^*, h(x^*,\xi))=((1+\alpha\xi), h((1+\alpha\xi),\xi))$ is a cusp singularity of codimension-$3$.
    \label{thm:BT2} 
\end{thm}
\begin{proof}
We begin by shifting coordinates via the affine transformation $x_1=x-x^*$ and $y_1=y-h(x^*,\xi)$, which brings the equilibrium $E^{**}=(x^*,h(x^*,\xi))$ to the origin. Expanding the system in a Taylor series around $E$, we obtain the following reduced system:
\begin{eqnarray}\label{eq:sys_1}
\left\{\begin{aligned}
\dot{x_1}= &\ y_1 +\frac{f(x^*,\xi)h_{xx}(x^*,\xi)}{2}x_1^2 + O(|(x_1,y_1)^3|),\\
\dot{y_1}=&\ \frac{\beta h(x^*,\xi) g_{xx}(x^*,\xi)}{2}x_1^2  + O(|(x_1,y_1)^3|)\\
\end{aligned}\right.
\end{eqnarray}
Under the near identity transformation $x_2=x_1$, and $y_2=\dot{x_1}$, system \eqref{eq:sys_1} reduces to
\begin{eqnarray}\label{eq:sys_2}
\left\{\begin{aligned}
\dot{x_2}= &\ y_2,\\
\dot{y_2}=&\ \zeta_1 x_2^2 + \zeta_2 x_2y_2  + O(|(x_2,y_2)^3|)\\
\end{aligned}\right.
\end{eqnarray}
Here, 
\begin{equation}
    \zeta_1=\frac{\beta h(x^*,\xi) g_{xx}(x^*,\xi)}{2}<0, \qquad \zeta_2=f(x^*,\xi) h_{xx}(x^*,\xi)
\end{equation}
Clearly, $\zeta_2=0$ if and only if $h_{xx}(x^*,\xi)=0$. We solve 
\begin{equation}
    h_{xx}(x^*,\xi)=h_{xx}\left(\frac{1}{\sqrt{b}},\xi\right)=-b (1+\alpha  \xi)-\sqrt{b}=0
\end{equation}
for $b$, gives
\begin{equation}
    b=\frac{1}{(1+\alpha  \xi)^2}:=b^*.
\end{equation}
If $b=b^*$, then we expand system \eqref{eq:sys_1} up to fourth order and reduce it to
\begin{eqnarray}\label{eq:sys_3}
\left\{\begin{aligned}
\dot{x}= &\ y,\\
\dot{y}=&\ x^2 + \frac{2 g^{(3)}(x^*,\xi)}{3\beta h(x^*,\xi)g_{xx}(x^*,\xi)^2}x^3+\frac{2(f(x^*,\xi)^2 h^{(3)}(x^*,\xi)-\beta g_{xx}(x^*,\xi))}{\beta^2 f(x^*,\xi) h(x^*,\xi)^2g_{xx}(x^*,\xi)^2}x^2y\\ 
&+ 4\frac{f_x(x^*,\xi)(4f(x^*,\xi)^2 h^{(3)}(x^*,\xi)+3\beta g_{xx}(x^*,\xi))-\beta f(x^*,\xi)g^{(3)}(x^*,\xi)}{3\beta^3 f(x^*,\xi)^2 h(x^*,\xi)^3g_{xx}(x^*,\xi)^3}x^3y\\
&+\frac{g^{(4)}(x^*,\xi)}{3\beta^2 h(x^*,\xi)^2g_{xx}(x^*,\xi)^3}x^4 +O(|(x,y)^5|)
\end{aligned}\right.
\end{eqnarray}
using the following near-identity transformation
\begin{equation}\label{eq:transf_2}
    \begin{aligned}
    x_1&=\frac{2x}{\beta h(x^*,\xi) h_{xx}(x^*,\xi)},\\
    y_1&=\frac{2}{\beta h(x^*,\xi) h_{xx}(x^*,\xi)}\left(y-\frac{2f(x^*,\xi) h^{(3)}(x^*,\xi)}{3\beta^2 h(x^*,\xi)^2g_{xx}(x^*,\xi)^2}x^3-\frac{4f_x(x^*,\xi) h^{(3)}(x^*,\xi)}{3\beta^3 h(x^*,\xi)^3g_{xx}(x^*,\xi)^3}x^4+ O(|(x,y)^5|)\right).
\end{aligned}
\end{equation}
Then by Proposition $5.3$ of \cite{lamontagne2008bifurcation}, system \eqref{eq:sys_3} is $C^\infty$ equivalent to the following
\begin{eqnarray}\label{eq:sys_4}
\left\{\begin{aligned}
\dot{x}= &\ y ,\\
\dot{y_1}=&\ x^2 + \zeta_3 x^3y + O(|(x,y)^5|)\\
\end{aligned}\right.
\end{eqnarray}
where
\begin{equation}
    \zeta_3=4\left(\frac{3(1+\alpha\xi)}{\beta(1+(\alpha-1)\xi)}\right)^3>0.
\end{equation}
Hence, it can be concluded that $E^{**}$ is a cusp singularity of codimension-$3$.
\end{proof}

\section{CRN - Deficiency}
\label{sec:crn}

CRNs are defined by a tuple of form $\mathcal{N}=(S, R)$ where $S$ is a finite set of species and $R$ is a finite set of reactions that operate on the species \cite{Aris1965_DCRNT, Feinberg1972_complex_balalnce_DCRNT, HornJackson1972_complex_balancing, Huang2019}. The semantics of a CRN are derived from nature based on reaction rates and reactant concentrations, and are commonly called the mass-action model \cite{Shapiro_MassActionGibbsFunction_DCRNT_EquilibriumComputation, Krambeck1970_DCRNT}. In this paper, we use a version of the mass-action model that characterizes the semantics (changes in species concentrations) as a system of polynomial autonomous differential equations \cite{Bournez2021_survey_analog, EpsteinPojman_CRN_book, computability_dcrn}. For example, consider a general $n-$species reaction network as follows:
\begin{equation}
    R_i:\qquad a_{i1}X_1+\dots+a_{in}X_n\xrightarrow{k_i}b_{i1}X_1+\dots+b_{in}X_n,\ \forall X_1,\dots,X_n\in S,\ \forall R_i\in R
\end{equation}

In the reaction network provided above, $k_i\in \mathbb{R}_{>0}$ is the rate constant of the $i^{th}$ reaction $R_i$. According to the mass-action kinetics, the rate equation induced by such a CRN is
\begin{equation}
    \dot{x_i}=\frac{dx_i}{dt}=\sum_{j=1}^{|R|}k_j(b_{ji}-a_{ji})\prod_{k=1}^{|S|}x_k^{a_{jk}},
\end{equation}
where $x_i$ is a function of time that represents the concentration of species $X_i\in S$ at time $t$.

Chemical Reaction Network (CRN) Theory introduces the idea of deficiency of reaction networks \cite{Feinberg1972_complex_balalnce_DCRNT, HornJackson1972_complex_balancing, Horn1972_deficiency_zero, Feinberg1995_deficiency_one, FEINBERG198059_deficiency_one, FEINBERG19872229_deficiency_one, feinberg2019foundations}. Deficiency $\delta$ of CRN is defined by the following equation
\begin{equation}
    \delta=n-l-s
    \label{eq:deficiency}
\end{equation}
Here, $n$ is the number of complexes taking part in the reaction network, $l$ is the number of linkage classes, and $s$ is the dimension of the stoichiometric subspace (SS). The SS is defined as
\begin{equation*}
    SS=span\{y'-y|y\rightarrow y'\in R,\ and\ y,y'\in\mathbb{R}_{\geq 0}^{|S|}\},
\end{equation*}
where $y,y'$ are complexes, $R$ is the set of reactions, and $|S|$ is the cardinality of the set of species $S$. In other words, deficiency intuitively measures how independent the dynamics of each reaction are from those of the others, given shared complexes and species.

Consider the following example given by Feinberg \cite{feinberg2019foundations} to explain each of the above-mentioned terminologies.
\begin{equation}
    \begin{split}
        R1&:\qquad A\leftrightharpoons 2B\qquad\quad\qquad\ \ R2:\qquad A+C\leftrightharpoons D\\
        R3&:\qquad D\rightarrow B+E\qquad\qquad R4:\qquad B+E\rightarrow A+C
    \end{split}
    \label{eq:CRN_example}
\end{equation}

There are a total of 5 species in the reaction network, which are $S=\{A, B, C, D, E\}$ and a set of 5 complexes, $\mathcal{C}=\{A, 2B, A+C, D, B+E\}$. Each complex $y\in\mathcal{C}$ is represented by a vector in $\mathbb{R}_{\geq 0}^{|S|}$ where $|S|$ is the cardinality of $S$, where each row represents the stoichiometric coefficient of an individual species appearing in the complex. For eg, considering the reaction $R1$ with complexes $y,y'$ of form $y\rightarrow y'$, the vectors are represented as $y=[1,0,0,0,0]^T,y'=[0,2,0,0,0]^T$ where the first, second, third, fourth, and fifth entities of the vectors represent the stoichiometric coefficients of $A, B, C, D,$ and $E$, respectively. The resultant vectors $y'-y$ for each reaction are called the stoichiometric vectors (SV) belonging to the SS. The term $s$ in equation \eqref{eq:deficiency} represents the dimension of this SS, i.e., the number of basis vectors that span the SS.

Now consider each complex as a node in a graph. Two nodes representing distinct complexes are connected by an undirected edge if there is a direct reaction between them. If we construct such a graph for the entire reaction network, we obtain a graph with connected subgraphs. We call each such connected subgraph a linkage class. The term $l$ in equation \eqref{eq:deficiency} represents the number of such linkage classes that are formed by the reaction network.

We can now analyze the deficiency of the CRN \eqref{eq:CRN_example}. Note that $R2, R3, R4$ form one linkage class while $R1$ alone forms one linkage class. As a result, we get $l=2$. The number of complexes $n=|\mathcal{C}|=5$. The stoichiometric vectors generated by the above CRN are
\begin{equation}
    R1:\begin{bmatrix}
        -1\\2\\0\\0\\0
    \end{bmatrix},\begin{bmatrix}
        1\\-2\\0\\0\\0
    \end{bmatrix}\qquad R2:\begin{bmatrix}
        -1\\0\\-1\\1\\0
    \end{bmatrix},\begin{bmatrix}
        1\\0\\1\\-1\\0
    \end{bmatrix}\qquad R3:\begin{bmatrix}
        0\\1\\0\\-1\\1
    \end{bmatrix}\qquad R4:\begin{bmatrix}
        1\\-1\\1\\0\\-1
    \end{bmatrix}
\end{equation}
From this, we can find that there exist the following 3 basis vectors implying $s=3$
\begin{equation}
    SS_{basis}=\left\{ \begin{bmatrix}
        -1\\2\\0\\0\\0
    \end{bmatrix},\quad \begin{bmatrix}
        1\\0\\1\\-1\\0
    \end{bmatrix},\quad \begin{bmatrix}
        0\\1\\0\\-1\\1
    \end{bmatrix} \right\}
\end{equation}
As a result, we get $\delta=5-2-3=0$.

In this regard, we analyze the deficiencies of various CRN systems representing ODE models, both with and without additional food (AF). The results of this analysis, along with the codims, are summarized in Table \ref{tab:sum_def}.

\begin{table}
\caption{Summary of Deficiency and BT-codimension with varying functional responses}
\label{tab:sum_def}
\centering
{%
\begin{tabular}{|c|C{2.55cm}|C{0.5cm}|C{1.5cm}|C{5cm}|C{0.5cm}|C{1.5cm}|}
\hline
& \multicolumn{3}{c|}{Without AF} &  \multicolumn{3}{c|}{With AF} \\
\hline
& Functional form & \centering $\delta$ & BT-codim & Functional form & \centering $\delta$ & BT-codim \\
\hline\hline
\multirow{2}{*}{\centering (i)}& \multirow{2}{*}{\centering $f(x) = x$}
& 2 & NA*
& $f(x,\xi,\alpha) = x$
& 3 & NA* \\ & 
& 
& \cite{zegeling2020singular}
& \vspace{1mm} $g(x,\xi,\alpha) = x+\xi$ & & \\
\hline

\multirow{2}{*}{\centering (ii)}& \multirow{2}{*}{\centering $f(x) = \frac{x}{1+ x}$}
& 5 & NA*
& $f(x,\xi,\alpha) = \frac{x}{1+\alpha \xi + x}$
& 6 & NA* \\ & 
& 
& \cite{cheng1981uniqueness, kuang1988uniqueness}
& \vspace{1mm} $g(x,\xi,\alpha) = \frac{\beta (x+\xi)}{1+\alpha \xi + x}$ & & \cite{SP07} \\
\hline

\multirow{2}{*}{\centering (iii)}& \multirow{2}{*}{\centering $f(x) = \frac{x^2}{ax^2+bx+ 1}$} & 5 & 3
& $f(x,\xi,\alpha) = \frac{x^2}{1+\alpha\xi^2+x^2}$ & 6 & NA* \\ & 
 & & \cite{lamontagne2008bifurcation} & \vspace{1mm} $g(x,\xi,\alpha) = \frac{\beta (x^{2}+\xi^{2})}{1+\alpha \xi^{2} + x^{2}}$ & & \cite{SPV18, ananth2021influence} \\
\hline

\multirow{3}{*}{\centering (iv)}& \multirow{3}{2.55cm}{\centering$f(x)$ as in (ii) $+$
\small Predator competition}
& 6 & 2 & \multirow{3}{5cm}{\centering $f(x,\xi,\alpha)$ and $g(x,\xi,\alpha)$ as in (ii) $+$ \small{Predator competition}} & 7 & 4 \\ & 
& & \vspace{1mm} \cite{lu2021global} & 
& & \vspace{1mm} \cite{kgoyal}\\  & & & & & & \\
\hline

\multirow{3}{*}{\centering (v)}& \multirow{3}{*}{\centering $f(x) = \frac{x}{ax^2+bx+ 1}$} & 8 & 3
&  $f(x,\xi,\alpha) = \frac{x}{(1+\alpha \xi)( \omega x^{2}+1) + x},$ & 10 & 3\\ & 
&  & \cite{zhu2003bifurcation, xiao2001global, xiao2006multiple, rothe1992multiple} & \vspace{1mm} $g(x,\xi,\alpha) = \frac{\beta (x+\xi(\omega x^{2}+1))}{(1+\alpha \xi)(\omega x^{2} + 1) + x}$ & & \\
\hline

\end{tabular}%
}

\vspace{0.1cm}
 \raggedright
 \textbf{Note: }{NA* describes that BT-bifurcation does not exist in the cited literature; only codimension-1 bifurcations (saddle-node, transcritical, Hopf bifurcations) or existence of limit cycles were investigated.}
\end{table}

\subsection{Deficiency calculation for Predator Competition--Type II}

\subsubsection{Without Additional Food}
\begin{equation} \label{eq:model_drift_cann_noajf}
\begin{aligned}
 \Dot{x} = x \left (1-\frac{x}{\gamma}\right )-\frac{x y}{1+x}\qquad\qquad
 \Dot{y} = \epsilon \left ( \frac{x}{1+x} \right)\ y  - k y - cy^2 
 \end{aligned}
 \end{equation}

 To make $\Dot{x},\Dot{y}$ represent the rate equation of a CRN, we considering $s=\frac{1}{1+x}$. Then equation \eqref{eq:model_drift_cann_noajf} becomes
 \begin{equation} \label{eq:model_drift_cann_noajf_substitute}
\begin{aligned}
 \Dot{x} = x-\frac{1}{\gamma}x^2-xys\qquad\qquad
 \Dot{y} = \epsilon xys  - k y - cy^2
 \end{aligned}
 \end{equation}
 
 Since $S$ is considered as a new species with species concentration $s$ that computes $\frac{1}{1+x}$, we now calculate $\frac{ds}{dt}$
 \begin{equation}
     \label{eq:s_derivative}
     \begin{split}
        \Dot{s}&=\frac{ds}{dt}=\frac{ds}{dx}\frac{dx}{dt}\qquad \bigg(where,\ \frac{ds}{dx}=-\bigg(\frac{1}{1+x}\bigg)^2=-s^2\bigg)\\
        \Rightarrow \Dot{s}&=-s^2x+\frac{1}{\gamma}s^2x^2+xys^2\qquad \bigg(substituting\ value\ of\ \frac{ds}{dx}\ and\ \frac{dx}{dt}\bigg)
     \end{split}
 \end{equation}

 The ODE system thus generated represents the following reaction network
 \begin{equation}
     \label{eq:crn_pctype2}
     \begin{split}
         2S+X\xrightarrow{1/2}X\qquad X\xrightarrow{1} 2X\qquad 2X\xrightarrow{1/2\gamma}\phi\qquad 2X+2S\xrightarrow{1/\gamma}2X+3S\qquad Y\xrightarrow{k}\phi\\
         X+Y+3S\xrightarrow{1}X+Y+4S\qquad X+Y+S\xrightarrow{1}Y+S\qquad X+Y+S\xrightarrow{\epsilon}X+S\qquad 2Y\xrightarrow{c/2}\phi
     \end{split}
 \end{equation}

 On closer analysis, we observe that $ n=13$, $ l=4$, and $s=3$. As a result, we get $\delta=6$

 \subsubsection{With Additional Food}
\begin{equation} \label{eq:model_drift_cannj}
\begin{aligned}
 \Dot{x} = x \left (1-\frac{x}{\gamma}\right )-\frac{x y}{1+x+\alpha \xi} \qquad\qquad
 \Dot{y} = \epsilon \left ( \frac{x+\xi}{1+x+\alpha \xi} \right)\ y  - k y - c \xi y^2
 \end{aligned}
 \end{equation}

As in the previous case, we consider $s=\frac{1}{1+x+\alpha\xi}$, which changes equation \eqref{eq:model_drift_cannj} as follows.
 \begin{equation} \label{eq:model_drift_cann_noajf_substitute_af}
\begin{aligned}
 \Dot{x} = x-\frac{1}{\gamma}x^2-xys\qquad\qquad
 \Dot{y} = \epsilon xys + \epsilon\xi ys  - k y - c\xi y^2
 \end{aligned}
 \end{equation}

 We now calculate $\Dot{s}$ as follows

 \begin{equation}
     \label{eq:s_derivative_af}
     \begin{split}
        \Dot{s}&=\frac{ds}{dx}\frac{dx}{dt}=-s^2x+\frac{1}{\gamma}s^2x^2+xys^2 \qquad \bigg(where,\ \frac{ds}{dx}=-\bigg(\frac{1}{1+x+\alpha\xi}\bigg)^2=-s^2\bigg)
     \end{split}
 \end{equation}

 This induces the following CRN
 \begin{equation}
     \label{eq:crn_pctype2_af}
     \begin{split}
        X\xrightarrow{1} 2X\qquad 2X\xrightarrow{1/2\gamma}\phi\qquad 2X+2S\xrightarrow{1/\gamma}2X+3S\qquad Y\xrightarrow{k}\phi\\
        X+Y+S\xrightarrow{1}Y+S\qquad X+Y+S\xrightarrow{\epsilon}X+S\qquad 2Y\xrightarrow{c\xi/2}\phi\\
        2S+X\xrightarrow{1/2}X\qquad X+Y+3S\xrightarrow{1}X+Y+4S\qquad Y+S\xrightarrow{\epsilon\xi}2Y+S
     \end{split}
 \end{equation}

 Analyzing the CRN \eqref{eq:crn_pctype2_af}, we get $n=14,l=4,$ and $s=3$ resulting in $\delta=7$.




\section{Discussion \& Conclusion}
\label{disc_conclusion}
In this study, we discuss the dynamics for a generalized additional food model \eqref{Eqn:1g} in which the predator $(y)$ growth not only depends on prey $(x)$ but also depends on the amount/quantity of additional food $(\xi)$. The quality of additional food is represented by $(\alpha)$, incorporating its nutritional value and the handling time required by predators. The model considers a general functional response $f(x, \xi, \alpha)$ and a numerical response $g(x, \xi, \alpha)$, which allows the analysis to apply to a broad class of AF models. We first showed that solutions remain nonnegative and bounded for all nonnegative initial conditions, thereby ensuring the model's well-posedness and biological feasibility (Theorems \ref{pos_inv_gen} and \ref{bdd_gen}). We then analyzed the existence and stability of the equilibrium points and derived conditions for their local and global stability.
The trivial equilibrium is always unstable, whereas the predator-free equilibrium $E_{\gamma}$ is stable only when predator mortality exceeds gains from consumption, else it is a saddle. For the coexistence equilibrium, a key requirement is that
$g_x(x^*,\xi,\alpha)>0$.
Under this condition, two sets of sufficient criteria were derived to ensure local stability. In contrast, when
$g_x(x^*,\xi,\alpha)<0$,
the coexistence equilibrium becomes a saddle (Theorems \ref{stability_trivial_equilibrium} and \ref{stability_coexistence_equilibrium}). \textcolor{black}{While local stability of the coexistence equilibrium can be determined by the sign of $g_x(x^*,\xi,\alpha)$, proving global stability is much more challenging. The asymmetry introduced by additional food in the functional and numerical responses prevents the direct use of many standard Lyapunov arguments. However, we were able to establish the global stability of the coexistence equilibrium by proving a sequence of lemmas and theorems, particularly see Lemmas \ref{lem:p1s_1}, \ref{lem:p2s_1} and Theorem \ref{thm:gs1_general}.}

The shape of the numerical response plays a key role in determining the system's long-term behavior.
As discussed in Remark \ref{Remark:unique}, if $g_x(x,\xi,\alpha) > 0$ then the model admits at most one coexistence equilibrium. This condition implies that predator growth increases with prey availability and with additional food, leading to a unique coexistence state whenever it exists. In contrast, if $g(x, \xi, \alpha)$ is non-monotonic, multiple coexistence equilibria may exist, suggesting that the interaction between predator growth and prey density is no longer straightforward. In such cases,  the equilibrium approached by the system can depend on the initial predator and prey densities. The existence of multiple coexistence equilibria may indicate richer dynamics, as changes in model parameters can alter both the number and stability of equilibria. This connection is evident in the Type IV functional response, where the non-monotonic nature of the response gives rise to a Bogdanov--Takens bifurcation and marks the onset of more complex dynamics.

\textcolor{black}{On the other hand, a CRN representation was formulated for these systems, providing a non-parametric lens to analyze stability and richness of dynamics. This work offers an analysis of the deficiency of reaction networks, an established metric in Chemical Reaction Network Theory. Although it does not clearly establish any relationship between deficiency and BT-codims, the analysis provides a different perspective on the richness of the dynamics of the studied systems. This is clearly indicated by the observation that all AF systems exhibit a higher deficiency than non-AF systems. This finding aligns with the theoretical and practical observations reported in this work and in the literature. In conclusion, we raise the following questions for future work.}

\begin{question*}
    In general, what interpretation of the richness of dynamics does the deficiency of the reaction network carry in systems like the one studied in this work?
\end{question*}

\begin{question*}
    Is there a defined relation between deficiency and BT-codim?
\end{question*}

\begin{question*}
    In general, can an increase in the value of deficiency be used as a strong indicator for richer dynamics of any dynamical system under study?
\end{question*}

To these ends, based on the computations herein, we make the following conjecture. 

\begin{cnj}
    Consider an AF model of general type \eqref{Eqn:type4_simp}, for which there exist multiple equilibrium. Let the deficiency of the AF model be $\delta_{1}$. Also consider the AF model with $\xi=0$,  whose deficiency is $\delta_{2}$. Then $\delta_{1} \geq \delta_{2}$. Furthermore, if the AF model has a BT bifurcation of codimension $\beta_{1}$, and the AF model with $\xi=0$ has a BT bifurcation of codimension $\beta_{2}$, then $\beta_{1} \geq \beta_{2}$.
\end{cnj}

\printbibliography
\end{document}